\documentclass[12pt]{article}

\usepackage{color}
\usepackage{graphicx}
\usepackage{capt-of}
\usepackage{float}
\usepackage{array}
\usepackage{eqnarray}
\usepackage{amssymb}
\usepackage{amsmath}
\usepackage{amsthm}
\usepackage{amsmath}
\usepackage{setspace}
\usepackage{enumerate}
\usepackage[french,english]{babel}
\usepackage{multirow}
\usepackage{caption}

\newtheorem{theorem}{Theorem}
\newtheorem{lemma}{Lemma}

\newtheorem*{Assumption}{Assumption A}

\title{Efficient inference about the tail weight in multivariate Student~$t$ distributions 
}
\author{Christophe Ley \, and \, Anouk Neven \\ 
 {\it \small ECARES and D\' epartement de Math\' ematique, Universit\'{e} libre de Bruxelles} \\ 
{\it \small UR en math\'ematiques, Universit\'e du Luxembourg}\\
} 
\date{}
\begin{document}
\maketitle

\begin{abstract}

 We propose a new testing procedure about the tail weight parameter of multivariate Student $t$ distributions by having recourse to the Le Cam methodology. Our test is asymptotically as efficient as the classical likelihood ratio test, but outperforms the latter by its flexibility and simplicity: indeed, our approach allows to estimate the location and scatter nuisance parameters by any root-$n$ consistent estimators, hereby avoiding numerically complex maximum likelihood estimation. The finite-sample properties of our test  are analyzed in a  Monte Carlo simulation study, and we  apply our method on a financial data set. We conclude the paper by indicating how to use this framework  for efficient point estimation.

\end{abstract}
\section{Introduction}\label{intro}

Under its most common form, the $k$-dimensional $t$ distribution admits the density
\begin{align} \label{student} f_{\pmb{\mu}, \pmb{\Sigma}, \nu}(\pmb{x}):=c_{\nu,k}|\pmb{\Sigma}|^{-1/2}\left(1+\|\pmb{\Sigma}^{-1/2}(\pmb{x}-\pmb{\mu})\|^2/\nu \right)^{-\frac{\nu+k}{2}}, \quad \pmb{x} \in \mathbb{R}^k,\end{align}
with location parameter $\pmb{\mu} \in \mathbb{R}^k$, scatter parameter $\pmb{\Sigma} \in S_k$, the class of symmetric and positive definite $k \times k$ matrices, and tail weight parameter $\nu \in \mathbb{R}_0^+$, and with normalizing constant 
$$c_{\nu,k}:=\frac{\Gamma \left(\frac{\nu+k}{2}\right)}{(\pi \nu)^{k/2} \,\Gamma \left(\frac{\nu}{2}\right)},$$
where the Gamma function is given by $\Gamma(z)=\int_{0}^\infty \exp(-t)t^{z-1}\, \mathrm{d}t$. Denoting by $F_{\pmb{\mu}, \pmb{\Sigma}, \nu}$ the cumulative distribution function, we deduce from (\ref{student}) that 
$$ F_{\pmb{\mu}, \pmb{\Sigma}, \nu}(\pmb{x})=O\left(\|\pmb{x}\|^{-\nu}\right)$$
as $\|\pmb{x}\| \to \infty$, which makes the Student $t$ distribution  a member of the class of heavy-tailed distributions with tail index $\gamma=\frac{1}{\nu}$. 
It follows that the kurtosis is regulated by the parameter $\nu$: the smaller $\nu$, the heavier the tails. For instance, for $\nu=1$, we retrieve the fat-tailed multivariate Cauchy distribution, whereas, when $\nu$ tends to infinity, we obtain the multivariate Gaussian distribution.

This heavy tail property is the key to the  burgeoning popularity of the $t$ distribution in empirical financial data modeling. Mandelbrot (1963) is the first to show that asset returns do not follow a Gaussian distribution, but have heavier tails. Ever since, there has been a great deal of empirical evidence supporting the existence of heavy-tailed models in finance and thereby challenging the well-established classical Gaussian assumption (see, \mbox{e.g.}, Fama, 1965 or Richardson and Smith, 1993). Among these models,  the class of $t$ distributions has been suggested as a tractable, more viable alternative, particularly because it captures the observed fat tails (see, \mbox{e.g.}, Praetz,~1972, Blattberg and Gonedes, 1974, Hagerman,~1978, Perry,~1983, Boothe and Glassman,~1987 or Kan and Zhou, 2003). For example, Kan and Zhou (2003) report that  the multivariate normality assumption on the distribution of Fama and French (1993)'s asset returns is rejected by a kurtosis test with a $p$-value of less than $0.01\%$. 

By advocating the use of the $t$ distribution instead of the Gaussian, a natural problem of interest consists in asking ``which Student distribution should be used?'' or, in statistical terms, in testing the null hypothesis $\mathcal{H}_0: \nu=\nu_0$ (where $\nu_0>0$ is fixed)  against alternatives of the form $\mathcal{H}_1^{\neq}: \nu \neq\nu_0$ (two-sided test),  $\mathcal{H}_1^{<}: \nu<\nu_0$ or $\mathcal{H}_1^{>}: \nu>\nu_0$ (one-sided tests). The likelihood ratio test (LR hereafter) provides the standard way to tackle this question. Clearly, the underlying test statistic invokes (i) the maximum likelihood estimators of  $(\pmb{\mu}, \pmb{\Sigma}, \nu)$ as solutions of the maximization of the log-likelihood function based on the $t$ density in (\ref{student}) without any constraint, and (ii) the maximum likelihood estimators of $(\pmb{\mu}, \pmb{\Sigma})$ subject to the restriction $\nu=\nu_0$. As is well-known, there exist no closed-form solutions to these maximization problems in the Student $t$ family and hence numerical procedures are required. A standard approach for solving numerically the likelihood equations is the popular EM algorithm of Dempster \emph{et al.}~(1977) or some variants of it discussed extensively for the Student~$t$ case in Liu and Rubin (1995) (see also the references therein). However, when $\nu$ is small or unknown, Liu and Rubin~(1995) have expressed the warning that maximum likelihood (ML) procedures might be misled due to numerous spikes with very high likelihood mass in the likelihood function, hereby rendering the associated parameter estimates ``of little practical interest''. See also the recent paper Gonz\'alez-Ar\'evalo and Pal (2013) where this problematic is treated.

Common large sample alternatives to the likelihood ratio test are the \textit{Wald test} (W) and the \emph{Lagrange multiplier test} (LM) (or \emph{Rao score test}), which is particularly popular in econometrics. Similar to the likelihood ratio test statistic, the underlying test statistics are based on maximum likelihood estimators: the Wald test requires the computation of the maximum likelihood estimator of the triple $(\pmb{\mu}, \pmb{\Sigma}, \nu)$, whereas the Lagrange multiplier test is derived from a constrained maximization problem, namely the maximization of the log-likelihood with respect to location and scatter subject to the constraint $\nu=\nu_0$. The LR, W and LM tests, regarded as the Holy Trinity in asymptotic statistics, are known to share the same efficiency properties (see Engle,~1984), but all three suffer from the non-existence of exact expressions and hence from the above-mentioned problems associated with ML-based methods  in the (univariate and) multivariate~$t$ case. 

In the present paper, we therefore propose a new technique for tackling hypothesis testing on the tail weight parameter $\nu$ under unspecified location and scatter. More precisely, we shall propose a class of testing procedures which are all asymptotically as powerful as the LR, W and LM tests, but improve on the latter by their flexibility and simplicity. As we shall see, our approach (described in Section~\ref{Intro2} below) will allow us to write out explicitly the powers of our tests against sequences of contiguous alternatives, which in general is extremely difficult to achieve with the classical tests. Moreover, the framework we develop here for hypothesis testing can be extended to other inferential issues such as parameter estimation.

\subsection{Optimal parametric tests via the Le Cam methodology}\label{Intro2}
As already mentioned above, the main purpose of the present work is to derive simple yet efficient tests for the tail weight parameter of multivariate Student $t$ distributions, more precisely, tests that are locally and asymptotically optimal. The underpinning optimality in this paper is the so-called \textit{maximin} optimality. Recall that a test $\phi^*$ is called maximin in the class $\mathcal{C}_\alpha$ of level-$\alpha$ tests for $\mathcal{H}_0$ against $\mathcal{H}_1$ if (i) $\phi^*$ has level $\alpha$ and (ii) the power of $\phi^*$ is such that 
$$\inf_{{\rm P} \in \mathcal{H}_1} \textnormal{E}_{{\rm P}}[\phi^*] \geq \sup_{\phi \in \mathcal{C}_\alpha} \inf_{{\rm P} \in \mathcal{H}_1} \textnormal{E}_{{\rm P}}[\phi].$$

The backbone of our construction will be the   Le Cam methodology. The concept of \textit{local asymptotic normality} (LAN) is among Le Cam's best-known contributions and plays an essential role in asymptotic optimality theory. To the best of our knowledge, nobody has yet taken advantage of the LAN approach in the framework of tail parameter inference for univariate and multivariate $t$ distributions. In order to ease the reading, we will briefly review here the LAN property and its contribution to the theory of hypothesis testing. The following definition of LAN corresponds to Le Cam and Yang (2000).

For all $n$, let $\mathcal{E}^{(n)}=\left( \mathcal{X}^{(n)},\mathcal{A}^{(n)}, \mathcal{P}^{(n)}:=\{{\rm P}^{(n)}_{\pmb{\theta}} | \pmb{\theta} \in \pmb{\Theta} \subset \mathbb{R}^k\} \right)$ be a sequence of $\pmb{\theta}$-parametric statistical models, called \textit{experiments} in Le Cam's terminology, and let $\delta_n$ be a sequence of positive numbers going to zero. The family $\mathcal{P}^{(n)}$ is called LAN at $\pmb{\theta} \in \pmb{\Theta}$ if there exists a sequence of random vectors $\pmb{\Delta}^{(n)}(\pmb{\theta})$, called \textit{central sequence}, and a non-singular symmetric matrix $\pmb{J}(\pmb{\theta})$, the associated \textit{Fisher information matrix},  such that, for every bounded sequence of vectors $\pmb{h}_n \in \mathbb{R}^k$,
\begin{align}\label{lan}\log \frac{d{\rm P}^{(n)}_{\pmb{\theta}+\delta_n\pmb{h}_n}}{d{\rm P}^{(n)}_{\pmb{\theta}}} -\pmb{h}_n' \pmb{\Delta}^{(n)}(\pmb{\theta})+\frac{1}{2}\pmb{h}_n'\pmb{J}(\pmb{\theta})\pmb{h}_n =o_{{\rm P}}(1) \end{align}
and  $\pmb{\Delta}^{(n)}(\pmb{\theta}) \stackrel{\mathcal{L}} \to \mathcal{N}_k({\pmb 0}, \pmb{J}(\pmb{\theta}))$, both under ${\rm P}^{(n)}_{\pmb{\theta}}$ as $n \to \infty$.

We easily see that the log-likelihood ratios $\log \frac{d{\rm P}^{(n)}_{\pmb{\theta}+\delta_n\pmb{h}_n}}{d{\rm P}^{(n)}_{\pmb{\theta}}}$ in (\ref{lan}) of a LAN family behave asymptotically like the log-likelihood ratio of the classical Gaussian shift experiment
$$\mathcal{E}_{\pmb{J}(\pmb{\theta})}=\left(\mathbb{R}^k, \mathcal{B}_k, \mathcal{P}_{\pmb{\theta}}:=\left\{ {\rm P}_{\pmb{h}, \pmb{\theta}}=\mathcal{N}_k\left( \pmb{J}(\pmb{\theta}) \pmb{h}, \pmb{J}(\pmb{\theta}\right))| \,\pmb{h} \in \mathbb{R}^k\right\}\right)$$
with a single observation which we denote as $\pmb{\Delta}$. This approximation of the statistical experiments $\mathcal{E}^{(n)}$ by the normal experiment $\mathcal{E}_{\pmb{J}(\pmb{\theta})}$ was  Le Cam's main motivation: \textit{the family of probability measures under study can be approximated very closely by a family of simpler nature} (Le Cam, 1960).

In the context of hypothesis testing, this approximation means that, asymptotically, all power functions that are implementable in the local experiments $\mathcal{E}^{(n)}$ are the power functions that are possible in the Gaussian shift experiment $\mathcal{E}_{\pmb{J}(\pmb{\theta})}$. In view of these considerations, it follows that asymptotically optimal tests in the local models can be derived by analyzing the Gaussian limit model. More precisely, if a test~$\phi(\pmb{\Delta})$ enjoys some exact optimality property in the Gaussian experiment  $\mathcal{E}_{\pmb{J}(\pmb{\theta})}$, then the corresponding sequence $\phi(\pmb{\Delta}^{(n)})$ inherits, locally and asymptotically, the same optimality properties in the sequence of experiments $\mathcal{E}^{(n)}$.

In the present work, we will investigate the LAN phenomenon  for a sequence of (univariate and multivariate) Student $t$ distributions with respect to the location, scatter and tail parameters (more precisely, we shall obtain the \emph{uniform} LAN property; see Section 2.2). As explained above, the (U)LAN property allows one to transfer the well-known optimal procedures from the classical Gaussian shift experiment to the $t$ model and thus paves the way towards the construction of locally and asymptotically optimal (in the maximin sense) tests for the tail parameter under unspecified location and scatter. As we shall see in the sequel, our test statistic strongly resembles that of the LM test, with the essential difference that our approach allows us to use \emph{any} root-$n$ consistent estimators of location and scatter under fixed $\nu=\nu_0$ (e.g., for $\nu_0>4$, we can even use the mean vector and variance-covariance matrix), hence circumvents their ML estimation. 
Furthermore, as already mentioned, the methodology we develop in this paper lends itself also quite well for parameter estimation, by having recourse to the celebrated one-step estimation method of Le Cam which renders \emph{any} root-$n$ consistent estimator as efficient as the ML estimator by adding to that estimator a certain form of the central sequence.

\subsection{Outline of the paper}
The paper is organized as follows. In Section 2, we set the notation and establish the   uniform local asymptotic normality property (with respect to the location,  scatter and tail parameters). In Section 3 we then construct the new optimal tests for tail weight. We  first explain in Section 3.1 how to estimate the nuisance parameters and identify then in Section 3.2  the efficient central sequence for tail weight under unknown location and scatter. In Section 3.3 we derive, thanks to the ULAN property, the locally and asymptotically optimal tests. We study the asymptotic behavior of our test statistic both under the null and under a sequence of local alternatives, allowing us to compute explicitly  the power of our test.   We explore in Section 4 the finite-sample properties  of our new tests via a  Monte Carlo simulation study. In Section 5, we apply our method on a financial data set. Section 6 provides an outlook on how to apply our approach in parameter estimation.  Finally, the Appendix contains the technical proofs.

\section{Uniform local asymptotic normality (ULAN)}

Throughout, the data points $\pmb{X}_1, \ldots,\pmb{ X}_n$ are assumed to follow a multivariate $t$ distribution with parameters $(\pmb{\mu},\pmb{\Sigma}, \nu)=:\pmb{\mathcal{V}}$. The relevant statistical experiment thus involves the parametric family
$$\mathcal{P}^{(n)}_{\pmb{\mathcal{V}}}=\{{\rm P}^{(n)}_{\pmb{\mathcal{V}}} | \,\pmb{\mathcal{V}} \in \mathbb{R}^k \times S_k \times \mathbb{R}^+_0\},$$
where ${\rm P}^{(n)}_{\pmb{\mathcal{V}}}$ stands for the joint distribution of $\pmb{X}_1, \ldots, \pmb{X}_n$.
The rest of this section is devoted to the establishment of the crucial  ULAN property of the considered parametric family. Indeed, as explained in the Introduction, this Taylor type expansion of the log-likelihood ratio is the backbone of our construction of optimal tests for the tail parameter $\nu$ under unspecified location $\pmb{\mu}$ and unspecified scatter $\pmb{\Sigma}$.

\subsection{Notation and definitions}
In order to ease readability, we start by explaining some notations that will be useful in the sequel.  Set $M_k$ the class of $k \times k$ matrices. We write $\textnormal{vec}(\pmb{A})$ for the $k^2$-vector obtained by stacking the columns of  a matrix $\pmb{A} \in M_k$ on top of each other, and $\textnormal{vech}(\pmb{A})$ for the $k(k+1)/2$- subvector of $\textnormal{vec}(\pmb{A})$ where only the upper diagonal entries in $\pmb{A}$ are considered. Define $\pmb{P}_k$ as the $k(k+1)/2\times k^2$ matrix such that $\pmb{P}_k'(\textnormal{vech}(\pmb{A}))=\textnormal{vec}(\pmb{A})$ for any $k \times k$ symmetric matrix $\pmb{A}$. Denoting by $\pmb{e}_\ell$ the $\ell$th vector in the canonical basis of $\mathbb{R}^k$ and by $\pmb{I}_k$ the $k\times k$ identity matrix, let
$$\pmb{K}_k := \sum_{i,j=1}^k (\pmb{e}_i \pmb{e}_j') \otimes (\pmb{e}_j \pmb{e}_i') \quad \textnormal{and} \quad \pmb{J}_k:= \sum_{i,j=1}^k (\pmb{e}_i \pmb{e}_j') \otimes (\pmb{e}_i\pmb{e}_j')= \textnormal{vec}(\pmb{I}_k) (\textnormal{vec}(\pmb{I}_k))',$$
where the $k^2 \times k^2$  matrix $\pmb{K}_k$ is known as the \textit{commutation matrix}. With this notation, $\pmb{K}_k (\textnormal{vec}(\pmb{A}))=\textnormal{vec}(\pmb{A}')$ and $\pmb{J}_k (\textnormal{vec}(\pmb{A}))=(\textnormal{tr}\pmb{A})(\textnormal{vec}(\pmb{I}_k))$. Finally, we write $\pmb{A}^{\otimes2}$ for the usual Kronecker product $\pmb{A} \otimes \pmb{A}$.

For multivariate $t$ distributions, we define the \textit{score vector} $\pmb{\mathcal{L}}_{\pmb{\mathcal{V}}}(\pmb{x}):=\left(\pmb{\mathcal{L}}_{\pmb{\mathcal{V}}}^{(i)}(\pmb{x})\right)_{i=1,2,3}$ for $\pmb{\mathcal{V}}$ as 
$$\pmb{\mathcal{L}}_{\pmb{\mathcal{V}}}(\pmb{x}):= \begin{pmatrix} 
\pmb{D}_{\pmb{\mu}} \log(f_{\pmb{\mu}, \pmb{\Sigma}, \nu}(\pmb{x}))\\
\pmb{D}_{\textnormal{vech}(\pmb{\Sigma})} \log(f_{\pmb{\mu}, \pmb{\Sigma}, \nu}(\pmb{x}))\\
D_\nu  \log(f_{\pmb{\mu}, \pmb{\Sigma}, \nu}(\pmb{x}))
\end{pmatrix},$$
where $D_{\textnormal{vech}(\pmb{\Sigma})}$ and $D_{\pmb{\mu}}$ are the usual gradients and $D_\nu$ is the usual derivative. Direct computations lead to
$$ \pmb{\mathcal{L}}_{\pmb{\mathcal{V}}}(\pmb{x})=\begin{pmatrix} \frac{1+k/\nu}{1+\|\pmb{\Sigma}^{-1/2}(\pmb{x}-\pmb{\mu})\|^2/\nu}\pmb{\Sigma}^{-1}(\pmb{x}-\pmb{\mu}) \\
\frac{1}{2} \pmb{P}_k (\pmb{\Sigma}^{\otimes2})^{-1/2} \textnormal{vec} \left( \frac{1+k/\nu}{1+\|\pmb{\Sigma}^{-1/2}(\pmb{x}-\pmb{\mu})\|^2/\nu}\pmb{\Sigma}^{-1/2}(\pmb{x}-\pmb{\mu})(\pmb{x}-\pmb{\mu})' \pmb{\Sigma}^{-1/2}-\pmb{I}_k \right)\\
\frac{c_{\nu,k}'}{c_{\nu,k}}-\frac{1}{2} \log \left( 1+\|\pmb{\Sigma}^{-1/2}(\pmb{x}-\pmb{\mu})\|^2/\nu\right)+ \frac{\nu+k}{2\nu^2} \frac{\|\pmb{\Sigma}^{-1/2}(\pmb{x}-\pmb{\mu})\|^2}{1+\|\pmb{\Sigma}^{-1/2}(\pmb{x}-\pmb{\mu})\|^2/\nu}
 \end{pmatrix},$$
where $c_{\nu,k}'$ stands for the derivative of the mapping $\nu\mapsto c_{\nu,k}$. Now, note that the $\pmb{\mu}$-score is anti-symmetric in $\pmb{x}-\pmb{\mu}$, while the scores for $\pmb{\Sigma}$ and $\nu$ are symmetric in $\pmb{x}-\pmb{\mu}$. Hence, the symmetry properties with respect to $\pmb{x}-\pmb{\mu}$ entail that the resulting Fisher information matrix $\pmb{\Gamma}(\pmb{\mathcal{V}})$, given by the covariance matrix of the score vector, partitions into
\begin{align} \label{fisher} \begin{pmatrix} 
\pmb{\Gamma}_{11}(\pmb{\mathcal{V}}) & \pmb{0} & \pmb{0} \\
\pmb{0} & \pmb{\Gamma}_{22}(\pmb{\mathcal{V}}) & \pmb{\Gamma}_{23}(\pmb{\mathcal{V}}) \\
\pmb{0} & (\pmb{\Gamma}_{23}(\pmb{\mathcal{V}}))' & \Gamma_{33}(\nu)
\end{pmatrix}.\end{align}
Lange \textit{et al.} (1989) derived explicit expressions for the entries of $\pmb{\Gamma}(\pmb{\mathcal{V}})$ given by
$$\pmb{\Gamma}_{11}(\pmb{\mathcal{V}}) = \frac{\nu+k}{\nu+k+2} \pmb{\Sigma}^{-1},$$
$$\pmb{\Gamma}_{22}(\pmb{\mathcal{V}})=\frac{1}{4}\pmb{P}_k (\pmb{\Sigma}^{\otimes2})^{-1/2} \left[ \frac{\nu+k}{\nu+k+2} (\pmb{I}_{k^2}+\pmb{K}_k+\pmb{J}_k)-\pmb{J}_k\right](\pmb{\Sigma}^{\otimes2})^{-1/2}\pmb{P}_k',$$
$$\pmb{\Gamma}_{23}(\pmb{\mathcal{V}})=\frac{-1}{(\nu+k+2)(\nu+k)} \pmb{P}_k (\pmb{\Sigma}^{\otimes2})^{-1/2} \textnormal{vec}(\pmb{I}_k),$$
and
$$\Gamma_{33}(\nu)=-\frac{1}{2} \left[ \frac{1}{2}\, \psi'\left(\frac{\nu+k}{2}\right)-\frac{1}{2}\,\psi'\left(\frac{\nu}{2}\right)+\frac{k}{\nu(\nu+k)}-\frac{1}{\nu+k}+\frac{\nu+2}{\nu(\nu+k+2)}\right],$$
where $\psi'$ is the trigamma function. Ley and Paindaveine (2010) recently proved that the information matrix is finite and non-singular for all $\pmb{\mathcal{V}} \in \mathbb{R}^k \times S_k \times \mathbb{R}^+_0$. 

Finally, for two parameter sets $\pmb{\mathcal{V}}_1=(\pmb{\mu}_1, \pmb{\Sigma}_1, \nu_1)$ and $\pmb{\mathcal{V}}_2=(\pmb{\mu}_2, \pmb{\Sigma}_2, \nu_2)$, we will make throughout the slight abuse of notation and write $\pmb{\mathcal{V}}_1+\pmb{\mathcal{V}}_2$ for $(\pmb{\mu}_1+\pmb{\mu}_2, \pmb{\Sigma}_1+\pmb{\Sigma}_2, \nu_1+\nu_2)$.

\subsection{Uniform local asymptotic normality (ULAN)}

With these notations and definitions in hand, we are ready to state the main technical result of this paper, namely the announced ULAN property of the multivariate $t$ family with respect to the location, scatter and tail parameters.
\begin{theorem}\label{ulan}
For any $\pmb{\mu} \in \mathbb{R}^k$, $\pmb{\Sigma} \in S_k$ and $\nu > 0$, the multivariate Student $t$ family $\mathcal{P}^{(n)}_{\pmb{\mathcal{V}}}$ is ULAN at $\pmb{\mathcal{V}}$, with central sequence
$$\pmb{\Delta}^{(n)}(\pmb{\mathcal{V}}):= \left(\pmb{\Delta}^{(n)}_i(\pmb{\mathcal{V}})\right)_{i=1,2,3}=\left(\frac{1}{\sqrt{n}} \sum_{k=1}^n \pmb{\mathcal{L}}_{\pmb{\mathcal{V}}}^{(i)}(\pmb{X}_k)\right)_{i=1,2,3},$$
and information matrix $\pmb{\Gamma}(\pmb{\mathcal{V}})$. More precisely, for any $\pmb{\mathcal{V}}^{(n)}=(\pmb{\mu}^{(n)}, \pmb{\Sigma}^{(n)}, \nu^{(n)})= \pmb{\mathcal{V}}+O(n^{-1/2})$ and for any bounded sequence $\pmb{\tau}^{(n)}=(\pmb{\tau}_1^{(n)}, \pmb{\tau}_2^{(n)}, \tau_3^{(n)}) \in \mathbb{R}^k \times M_k \times \mathbb{R}$ such that $\pmb{\Sigma}^{(n)}+n^{-1/2}\pmb{\tau}^{(n)}_2\in S_k$ and $\nu^{(n)}+n^{-1/2}\tau_3^{(n)}>0$, we have
\begin{align*} \Lambda^{(n)}_{\pmb{\mathcal{V}}^{(n)}+n^{-1/2} \pmb{\tau}^{(n)}/ \pmb{\mathcal{V}}^{(n)}} (\pmb{X}_1, \ldots, \pmb{X}_n):=& \log \left( \mathrm{d}{\rm P}^{(n)}_{\pmb{\mathcal{V}}^{(n)}+n^{-1/2}\pmb{\tau}^{(n)}}/ \mathrm{d}{\rm P}^{(n)}_{\pmb{\mathcal{V}}^{(n)}}\right)\\
=& (\mathcal{T}^{(n)})'\pmb{\Delta}^{(n)}\left(\pmb{\mathcal{V}}^{(n)}\right)-\frac{1}{2} (\mathcal{T}^{(n)})' \pmb{\Gamma}(\pmb{\mathcal{V}})\mathcal{T}^{(n)}+o_{\rm P}(1),
\end{align*}
where $\mathcal{T}^{(n)}:=((\pmb{\tau}_1^{(n)})',({\rm vech}(\pmb\tau_2^{(n)}))',\tau_3^{(n)})'$, and $\pmb{\Delta}^{(n)}\left(\pmb{\mathcal{V}}^{(n)}\right) \stackrel{\mathcal{L}}{\longrightarrow} \mathcal{N}_{k+k(k+1)/2+1}\left(\pmb{0}, \pmb{\Gamma}(\pmb{\mathcal{V}}) \right)$, both under ${\rm P}_{\pmb{\mathcal{V}}^{(n)}}^{(n)}$, as $n \to \infty$.
\end{theorem}

\begin{proof} 
The proof is quite immediate since we are not working within a semi-parametric family of distributions (hence we do not have to deal with an infinite-dimensional nuisance parameter); the problem considered indeed involves a parametric family of distributions with densities meeting the most classical regularity conditions.
 In particular, one readily obtains that (i) $(\pmb{\mu}', (\textnormal{vech}(\pmb{\Sigma}))', \nu)' \mapsto f^{1/2}_{\pmb{\mu}, \pmb{\Sigma}, \nu}(\pmb{x})$ is continuously differentiable for every $\pmb{x}\in\mathbb{R}^k$ and (ii) the associated Fisher information matrix is well defined and continuous in $(\pmb{\mu}', (\textnormal{vech}(\pmb{\Sigma}),\nu)'$. Thus, by Lemma 7.6 of van der Vaart (1998), $(\pmb{\mu}', (\textnormal{vech}(\pmb{\Sigma}))', \nu)' \mapsto f^{1/2}_{\pmb{\mu}, \pmb{\Sigma}, \nu}(\pmb{x})$ is differentiable in quadratic mean, and the ULAN property follows from Theorem 7.2 of van der Vaart (1998). This completes the proof.
 \end{proof}

 Note that the term \og uniform" in ULAN  indicates that the local asymptotic normality  property holds not only at $\pmb{\mathcal{V}}$, but in a neighborhood of that point. Further,  the ULAN structure entails the following \textit{asymptotic linearity} property of the central sequence:
\begin{align}\label{uniform} \pmb{\Delta}^{(n)}\left(\pmb{\mathcal{V}}+n^{-1/2} \pmb{\tau}^{(n)}\right)=\pmb{\Delta}^{(n)}(\pmb{\mathcal{V}})-\pmb{\Gamma}(\pmb{\mathcal{V}})\mathcal{T}^{(n)}+o_{{\rm P}}(1),\end{align}
as $n \to \infty$, under ${\rm P}_{\pmb{\mathcal{V}}}^{(n)}$. As we shall see in the next section, the preceding asymptotic linearity property is needed, first to construct the optimal tests for tail weight under unspecified location and scatter and second to derive explicit expressions of the power function for those tests.

\section{Locally and asymptotically optimal tests for tail weight}\label{three}

In this section, we make use of the ULAN property established in the previous section in order to construct locally and asymptotically  optimal tests for the tail parameter $\nu$. To this end, we apply the  Le Cam methodology \,for parametric tests to the present context. We shall proceed in three steps: first, we explain how to estimate the nuisance parameters $\pmb{\mu}$ and $\pmb{\Sigma}$ (Section~\ref{threeone}), second we show how to modify the central sequence $\pmb{\Delta}^{(n)}$ in order to take into account the asymptotic correlation between the scatter and tail parameters (Section~\ref{threetwo}), and finally we write out the test statistic and study its asymptotic properties (Section~\ref{threethree}).

\subsection{Estimation of the nuisance parameters}\label{threeone}

Clearly, it is hard to think of any practical problem where the location and scatter are specified. We thus concentrate on asymptotic optimality under unspecified location and scatter: $\pmb{\mu}$ and $\pmb{\Sigma}$ play the roles of nuisance parameters, whereas $\nu$ is the parameter of interest. In particular, the central sequence $\pmb{\Delta}^{(n)}(\pmb{\mathcal{V}}^{(n)})$ given in Theorem \ref{ulan}  (or a modified version of it, see the end of Section~\ref{threetwo}) depends on the unknown values of $\pmb{\mu}$ and $\pmb{\Sigma}$, hence is not yet a true statistic. This problem can be solved by replacing $(\pmb{\mu}, \pmb{\Sigma})$ with an adequate estimator $(\pmb{\hat{\mu}}^{(n)}, \pmb{\hat{\Sigma}}^{(n)})$ whilst, of course, paying attention to the asymptotic effects of such a substitution.  It is precisely here that the asymptotic linearity property comes in handy: our aim consists in using relation (\ref{uniform}) with $\pmb{\tau}^{(n)}=\left(n^{1/2}(\pmb{\hat{\mu}}^{(n)}-\pmb{\mu}), n^{1/2}(\pmb{\hat{\Sigma}}^{(n)}-\pmb{\Sigma}), 0\right)$, providing the asymptotic link between $\pmb{\Delta}^{(n)}\left(\hat{\pmb{\mu}}^{(n)},\hat{\pmb{\Sigma}}^{(n)},\nu\right)$ and $\pmb{\Delta}^{(n)}\left({\pmb{\mu}},{\pmb{\Sigma}},\nu\right)$. However, this replacement is not straightforward and imposes one more condition on the estimators, which is summarized in the following assumption which we state for a general parametric model $\mathcal{P}^{(n)}_\lambda$ (see Kreiss 1987, where such replacements have been worked out in detail).

\begin{Assumption} The sequence of estimators $\hat{\lambda}^{(n)}$ defined for a sequence of experiments $\mathcal{P}_{\lambda}^{(n)}=\{ {\rm P}_\lambda^{(n)} |\, \lambda \in~\Lambda\}$ indexed by a parameter $\lambda$ belonging to the parameter space $\Lambda$ is 
\begin{enumerate}[(i)]
\item root-$n$ consistent; that is, $n^{1/2}(\hat{\lambda}^{(n)}-\lambda)=O_{{\rm P}}(1)$ under ${{\rm P}}^{(n)}_{\lambda}$ for all $\lambda \in \Lambda$;
\item locally and asymptotically discrete; that is, the number of possible values of $\hat{\lambda}^{(n)}$ in $\lambda$-centered balls with $O(n^{-1/2})$ radius is uniformly bounded, as $n \to \infty$.
\end{enumerate}
\end{Assumption}
\noindent Both estimators $\hat{\pmb{\mu}}^{(n)}$ and $\hat{\pmb{\Sigma}}^{(n)}$ will have to satisfy this requirement in what follows. Local asymptotic discreteness is a concept that goes back to Le Cam (1986) and is quite standard in parameter estimation, since it turns root-$n$ consistent estimators into uniformly root-$n$ consistent ones (see Lemma 4.4 in Kreiss,~1987). Denoting by $\lceil x \rceil$ the smallest integer larger than or equal to $x$ and by $c_0$ an arbitrary positive constant that does not depend on $n$, any sequence of estimators $\hat{\lambda}^{(n)}$ of $\lambda$ can be discretized by replacing it with
$$\hat{\lambda}^{(n)}_{\sharp}:=c_0^{-1}n^{-1/2}\textnormal{sign}\left(\hat{\lambda}^{(n)}\right)\left \lceil  c_0n^{1/2}\hat{\lambda}^{(n)}\right\rceil.$$
In practice, however, such discretization is not required, as $c_0$ can be chosen large enough to make discretization be irrelevant at the fixed sample size $n$. Assumption A(ii) is thus a purely technical requirement with little practical implications, so that the preliminary estimator essentially only needs to be consistent at the standard root-$n$ rate.  
Obvious examples of such estimators for $\pmb{\mu}$ and $\pmb{\Sigma}$ are of course the sample mean $\pmb{\bar{X}}$ and (a multiple of) the sample covariance matrix ${\pmb S}(1-2/\nu_0)$ (with $\nu_0$ the value under the null). However, it is well-known that the latter require first respectively second moments in order to exist, and second respectively fourth moments in order to be root-$n$ consistent. In the univariate case, simple moment-free estimators for location and scale are the median and the median of absolute deviations (MAD). In higher dimensions, if moment conditions are to be avoided, one can for example use as $\hat{\pmb{\mu}}^{(n)}$ the spatial median of M\"ott\"onen and Oja (1995) and  as $\hat{\pmb{\Sigma}}^{(n)}$ Tyler~(1987)'s shape estimator (adjusted to be a scatter estimator by multiplication of a scale estimator), or construct location and scatter estimators  via the Minimum Covariance Determinant (MCD) method, the root-$n$ consistency of which can be found in Cator and Lopuha\"a~(2012). This underlines one of the advantages of our proposal: the user can freely choose his/her  preferred root-$n$ consistent location and scatter estimators (according to the needs of a given situation),  and is not forced to use the more complicated maximum likelihood estimators for location and scatter.

\subsection{An efficient central sequence for tail weight}\label{threetwo}

The block-diagonal structure of the Fisher information matrix (\ref{fisher}) confirms that the blocks $\pmb{\mu}$ and $(\pmb{\Sigma}, \nu)$ are asymptotically uncorrelated and thus the non-specification of  $\pmb{\mu}$ does not affect, asymptotically, inferential procedures for $\nu$ and/or $\pmb\Sigma$. More precisely, it follows from the ULAN property that replacing $\pmb{\mu}$ with an estimator $\pmb{\hat{\mu}}^{(n)}$ satisfying Assumption~A  has no influence, asymptotically, on $\Delta^{(n)}_3(\pmb{\mathcal{V}})$, the $\nu$-part of the central sequence. This can be seen via the asymptotic linearity (\ref{uniform}) of the central sequence combined with Lemma 4.4 in Kreiss (1987) (if unclear, see the proof of Lemma \ref{linea}, where we develop this argument). On the contrary, a non-zero asymptotic covariance $\pmb{\Gamma}_{23}(\pmb{\mathcal{V}})$ between the scatter and the tail part of the central sequence implies that the cost of not knowing the actual value of $\pmb{\Sigma}$ is strictly positive when performing inference on $\nu$.  This means that a local perturbation of $\pmb{\Sigma}$ has the same asymptotic impact on $\Delta^{(n)}_3(\pmb{\mathcal{V}})$ as a local perturbation of $\nu$. This impact will be taken into account in what follows.

The ULAN structure and the convergence of the local experiments to the Gaussian shift experiment
$$\begin{pmatrix} \pmb{\Delta}_2 \\ \Delta_3 \end{pmatrix} \sim \mathcal{N}_{k(k+1)/2+1}\left( \begin{pmatrix} \pmb{\Gamma}_{22}(\pmb{\mathcal{V}})  &  \pmb{\Gamma}_{23}(\pmb{\mathcal{V}}) \\  (\pmb{\Gamma}_{23}(\pmb{\mathcal{V}}))' &  \Gamma_{33}(\nu)\end{pmatrix}\begin{pmatrix} {\rm vech}(\pmb{\tau}_2) \\ \tau_3 \end{pmatrix}, \;  \begin{pmatrix} \pmb{\Gamma}_{22}(\pmb{\mathcal{V}})  &  \pmb{\Gamma}_{23}(\pmb{\mathcal{V}}) \\  (\pmb{\Gamma}_{23}(\pmb{\mathcal{V}}))' &  \Gamma_{33}(\nu)\end{pmatrix}\right),$$
where $(\pmb{\tau}_2, \tau_3) \in M_k\times \mathbb{R}$, imply that 
locally and asymptotically optimal inference on $\nu$ should be based on the residual of the regression of $\Delta_3$ with respect to $\pmb{\Delta}_2$, computed at $\Delta^{(n)}_3(\pmb{\mathcal{V}})$ and $\pmb{\Delta}^{(n)}_2(\pmb{\mathcal{V}})$. The resulting \textit{efficient central sequence} for $\nu$ thus takes the form
$$\Delta_3^{(n)*}(\pmb{\mathcal{V}}):=\Delta^{(n)}_3(\pmb{\mathcal{V}})-(\pmb{\Gamma}_{23}(\pmb{\mathcal{V}}))' (\pmb{\Gamma}_{22}(\pmb{\mathcal{V}}))^{-1} \pmb{\Delta}^{(n)}_2(\pmb{\mathcal{V}}).$$
 The projection of $\Delta_3^{(n)}(\pmb{\mathcal{V}})$ onto the subspace orthogonal to $ \pmb{\Delta}_2^{(n)}(\pmb{\mathcal{V}})$ ensures that the new efficient central sequence is asymptotically uncorrelated with the central sequences corresponding to $\pmb{\mu}$ and $\pmb{\Sigma}$.
Under ${\rm P}_{\pmb{\mathcal{V}}}^{(n)}$, the efficient central sequence for tail weight  is asymptotically normal, with mean zero and covariance (\textit{efficient Fisher information for tail weight}) 
\begin{align*}\Gamma_{33}^*(\pmb{\mathcal{V}})&:=\Gamma_{33}(\nu)-(\pmb{\Gamma}_{23}(\pmb{\mathcal{V}}))' (\pmb{\Gamma}_{22}(\pmb{\mathcal{V}}))^{-1}\pmb{\Gamma}_{23}(\pmb{\mathcal{V}}),
\end{align*}
which is non-zero (see Theorem 4.2 from Ley and Paindaveine, 2010). In the one-dimensional case, the efficient central sequence and Fisher information for $\nu$ reduce after some elementary calculations to
\begin{align*}\Delta_3^{(n)*}(\pmb{\mathcal{V}})=&\frac{1}{\sqrt{n}} \sum_{i=1}^n \left\{ \frac{\nu+3}{2\nu^2} \left(\frac{X_i-\mu}{\sigma}\right)^2\left(1+\frac{(X_i-\mu)^2}{\sigma^2\nu}\right)^{-1} -\frac{1}{2}  \log \left( 1+\frac{(X_i-\mu)^2}{\sigma^2 \nu}\right)-\frac{c'_\nu}{c_\nu}-\frac{1}{\nu(\nu+1)}\right\}\end{align*}
and
\begin{align*}\Gamma_{33}^*(\nu)= -\frac{1}{2}\left[ \frac{1}{2}\, \psi'\left(\frac{\nu+1}{2}\right)-\frac{1}{2}\,\psi'\left(\frac{\nu}{2}\right)+\frac{\nu+3}{\nu(\nu+1)^2}\right],
\end{align*}
where $\sigma^2:=\Sigma$. In particular, if $k=1$, $\Gamma_{33}^*(\nu)$ does not depend on the location nor on the scale (whence the notation). In the following result (whose proof is deferred to the Appendix), we derive the asymptotic linearity property of the efficient central sequence for tail weight.

\begin{lemma} \label{linea}
For any  $\pmb{\mu} \in \mathbb{R}^k$, $\pmb{\Sigma} \in S_k$ and $\nu>0$, and for any bounded sequence $\pmb{\tau}^{(n)}=\left(\pmb{\tau}_1^{(n)}, \pmb{\tau}_2^{(n)}, \tau_3^{(n)}\right) \in \mathbb{R}^k \times M_k\times \mathbb{R}$ such that $\pmb{\Sigma}+n^{-1/2}\pmb{\tau}^{(n)}_2\in S_k$ and $\nu+n^{-1/2}\tau_3^{(n)}>0$, we have that, under ${\rm P}_{\pmb{\mathcal{V}}}^{(n)}$ and as~$n\rightarrow\infty$,
$$\Delta_3^{(n)*}\left(\pmb{\mathcal{V}}+n^{-1/2}\pmb{\tau}^{(n)}\right)=\Delta_3^{(n)*}(\pmb{\mathcal{V}})-\Gamma_{33}^*(\pmb{\mathcal{V}})\tau_3^{(n)}+o_{{\rm P}}(1).$$
In particular, if $\tau_3^ {(n)}=0$, 
$$\Delta_3^{(n)*}\left(\pmb{\mu}+n^{-1/2}\pmb{\tau}_1^{(n)}, \pmb{\Sigma}+n^{-1/2}\pmb{\tau}_2^{(n)}, \nu\right)=\Delta_3^{(n)*}\left(\pmb{\mu}, \pmb{\Sigma}, \nu \right)+o_{{\rm P}}(1).$$
\end{lemma}

It is precisely here that Assumption A (and especially the local asymptotic discreteness) comes into play: for any estimators $\pmb{\hat{\mu}}^{(n)}$  and $\pmb{\hat{\Sigma}}^{(n)}$ satisfying Assumption A, Lemma 4.4 in Kreiss (1987) ensures that the asymptotic linearity property of Lemma \ref{linea} holds after replacement of $\left(\pmb{\tau}_1^{(n)}, \pmb{\tau}_2^{(n)},0\right)$ by the random quantity $\left( n^{1/2}(\pmb{\hat{\mu}}^{(n)}-\pmb{\mu}), n^{1/2}\left(\pmb{\hat{\Sigma}}^{(n)}-\pmb{\Sigma}\right),0\right)$, which entails that, asymptotically under ${\rm P}_{\pmb{\mathcal{V}}}^{(n)}$,
\begin{align} \label{lineaeff} \Delta_3^{(n)*}\left(\pmb{\hat{\mu}}^{(n)},\pmb{\hat{\Sigma}}^{(n)}, \nu \right)=\Delta_3^{(n)*}\left(\pmb{\mu}, \pmb{\Sigma}, \nu \right)+o_{{\rm P}}(1).\end{align}
The latter (asymptotic) equality in probability will allow us to derive the asymptotic behavior of our optimal tests for tail weight in the next section.

\subsection{Simple optimal tests for tail weight}\label{threethree}

As described in the Introduction, the ULAN property allows us to translate  optimal procedures from Gaussian shift experiments into our Student $t$ model. This, in combination with the developments of the previous section, entails that the optimal test $\phi^{(n)}_{\nu_0}$ for $\mathcal{H}_0: \nu =\nu_0$ (with $\nu_0>0$  fixed) in the Student $t$ family with unspecified location $\pmb{\mu}$ and scatter $\pmb{\Sigma}$ should be based on the efficient central sequence for tail weight. More concretely, $\phi^{(n)}_{\nu_0}$ rejects the null (at asymptotic level $\alpha$) in favor of $\mathcal{H}_1^{\neq}: \nu \neq \nu_0$ whenever the test statistic $|Q^{(n)}_{\nu_0}|$, with
$$Q^{(n)}_{\nu_0}:=\frac{\Delta_3^{(n)*}\left(\pmb{\hat{\mu}}^{(n)}, \pmb{\hat{\Sigma}}^{(n)}, \nu_0\right)}{\sqrt{\Gamma_{33}^*\left(\pmb{\hat{\mu}}^{(n)}, \pmb{\hat{\Sigma}}^{(n)}, \nu_0\right)}}$$
where the estimators $\pmb{\hat{\mu}}^{(n)}$ and $\pmb{\hat{\Sigma}}^{(n)}$  satisfy Assumption A, exceeds  $z_{\alpha/2}$, the $\alpha/2$-upper quantile of a standard Gaussian distribution. Thanks to (\ref{lineaeff}), we can derive  the asymptotic properties of $Q^{(n)}_{\nu_0}$, and hence also of $\phi^{(n)}_{\nu_0}$, in the next theorem (see the Appendix for a proof).

\begin{theorem}\label{maximin}
Fix $\nu_0>0$ and suppose that $\hat{\pmb{\mu}}^{(n)}$ and $\hat{\pmb{\Sigma}}^{(n)}$ satisfy Assumption~A. Then \begin{enumerate}[(i)]
\item $Q^{(n)}_{\nu_0}$ is asymptotically standard normal under $\bigcup_{\pmb{\mu} \in \mathbb{R}^k} \bigcup_{\pmb{\Sigma}\in S_{k}}{\rm P}^{(n)}_{(\pmb{\mu}, \pmb{\Sigma}, \nu_0)}$;
\item $Q^{(n)}_{\nu_0}$ is asymptotically  normal with~mean~$\tau_3\sqrt{\Gamma_{33}^*(\pmb{\mu}, \pmb{\Sigma}, \nu_0)}$~and variance 1 under $\bigcup_{\pmb{\mu} \in \mathbb{R}^k} \bigcup_{\pmb{\Sigma}\in S_{k}}{\rm P}^{(n)}_{(\pmb{\mu}, \pmb{\Sigma}, \nu_0+n^{-1/2}\tau_3^{(n)})}$, where  $\tau_3^{(n)} \in \mathbb{R}$ is a bounded sequence satisfying $\nu_0+n^{-1/2}\tau_3^{(n)}>0$ and $\tau_3:=\lim_{n \to \infty} \tau_3^{(n)}$;
\item the sequence of tests $\phi^{(n)}_{\nu_0}$ has asymptotic level $\alpha$ under $\mathcal{H}_0:=\bigcup_{\pmb{\mu} \in \mathbb{R}^k} \bigcup_{\pmb{\Sigma}\in S_{k}}{\rm P}^{(n)}_{(\pmb{\mu}, \pmb{\Sigma}, \nu_0)}$ and is locally and asymptotically maximin for testing $\mathcal{H}_0$ against  \\$\mathcal{H}_1^{\neq}:=\bigcup_{\pmb{\mu} \in \mathbb{R}^k} \bigcup_{\pmb{\Sigma}\in S_{k}}\bigcup_{0<\nu\neq\nu_0}{\rm P}^{(n)}_{(\pmb{\mu}, \pmb{\Sigma}, \nu)}$.
\end{enumerate}
\end{theorem}

The corresponding one-sided tests are easily derived along the same lines. Theorem~\ref{maximin} shows that our test has the same asymptotic behavior as the LR, W and LM tests. As already explained in the Introduction, our tests improve on these classical proposals by the non-necessity of estimating $\nu$ (neither under the null nor under the alternative), the freedom of choice among root-$n$ consistent estimators $\hat{\pmb{\mu}}^{(n)}$ and $\hat{\pmb{\Sigma}}^{(n)}$ and the ensuing simplicity. Yet another advantage of our Le Cam approach lies in the fact that Part~(ii) of Theorem \ref{maximin} makes it possibly to easily write down the power of $\phi^{(n)}_{\nu_0}$.  Denoting by $\Phi$ the cumulative distribution function of the standard Gaussian distribution, the asymptotic power of  $\phi^{(n)}_{\nu_0}$ under local alternatives of the form $\bigcup_{\pmb{\mu} \in \mathbb{R}^k} \bigcup_{\pmb{\Sigma} \in S_{k}}{ \rm P}^{(n)}_{(\pmb{\mu}, \pmb{\Sigma}, \nu_0+n^{-1/2}\tau_3^{(n)})}$  ($\tau_3:=\lim_{n \to \infty} \tau_3^{(n)}$) is then given by
$$1-\Phi\left(z_{\alpha/2}-\tau_3\sqrt{\Gamma_{33}^*(\pmb{\mu}, \pmb{\Sigma}, \nu_0)}\right)+\Phi\left(-z_{\alpha/2}-\tau_3\sqrt{\Gamma_{33}^*(\pmb{\mu}, \pmb{\Sigma}, \nu_0)}\right),$$
and by 
$$1-\Phi\left(z_\alpha-\tau_3\sqrt{\Gamma_{33}^*(\pmb{\mu}, \pmb{\Sigma}, \nu_0)}\right) \qquad \textnormal{and}\qquad \Phi\left(-z_\alpha-\tau_3\sqrt{\Gamma_{33}^*(\pmb{\mu}, \pmb{\Sigma}, \nu_0)}\right)$$
in the respective one-sided tests  against $\mathcal{H}_1^{>}: \nu>\nu_0$ and $\mathcal{H}_1^{<}: \nu <\nu_0$.  

We conclude this section by briefly quantifying the loss in power due to the estimation of the scatter parameter. The ULAN result in Theorem \ref{ulan} is about the \og unspecified scatter \fg \,model, but it evidently  entails ULAN for specified scatter.   Thus, locally and asymptotically optimal tests for $\mathcal{H}_0: \nu=\nu_0$ under specified scatter $\pmb{\Sigma}$ reject $\mathcal{H}_0$ (at asymptotic level $\alpha$) whenever
$$|Q^{(n)}_{\nu_0,\pmb{\Sigma}}|>z_{\alpha/2},$$
with
$$Q^{(n)}_{\nu_0,\pmb{\Sigma}}:=\frac{\Delta^{(n)}_3\left(\pmb{\hat{\mu}}^{(n)}, \pmb{\Sigma}, \nu_0\right)}{\sqrt{\Gamma_{33}\left(\nu_0\right)}},$$
where $\pmb{\hat{\mu}}^{(n)}$ is a sequence of estimators satisfying Assumption A. Along the same lines as in the proof of Theorem \ref{maximin}, one can show that the asymptotic behavior of $Q_{\nu_0,\pmb{\Sigma}}^{(n)}$ under the local alternatives is $\mathcal{N}(\tau_3 \sqrt{\Gamma_{33}(\nu_0)},1).$ The non-centrality parameters in the asymptotic non-null distributions of $Q^{(n)}_{\nu_0}$ and $Q_{\nu_0,\pmb{\Sigma}}^{(n)}$ allow for computing the efficiency loss due to an unspecified $\pmb{\Sigma}$, which is simply the difference of those local shifts, hence
$$\tau_3\left(\sqrt{\Gamma_{33}(\nu_0)}-\sqrt{\Gamma_{33}^*(\pmb{\mu}, \pmb{\Sigma}, \nu_0)}\right).$$
The positive definiteness of $\pmb{\Gamma}_{22}(\pmb{\mathcal{V}})$ in $\Gamma_{33}^*(\pmb{\mu}, \pmb{\Sigma}, \nu_0)$ confirms the unsurprising fact that this loss is strictly positive. Quite remarkably, it does not depend on the scale $\sigma:=\sqrt{\pmb{\Sigma}}$ in the one-dimensional setup.

{\section{Monte Carlo simulation study}\label{simus}

In order to investigate the finite-sample properties of the optimal test $\phi^{(n)}_{\nu_0}$ proposed in this paper, we have conducted a Monte Carlo simulation study, whose code  has been written in $\mathtt{R}$ and is available from the authors upon request.

More concretely, we compare the power of our test $\phi^{(n)}_{\nu_0}$ to that of the likelihood ratio test LR
in dimension $k=6$, for the two sample sizes  $n=200$ and 500 and for the two null values $\nu_0=5$ and $\nu_0=10$. To this end, we have for each setting generated $N=2,500$ independent samples of $6$-variate Student $t$ random vectors with $\pmb\mu=\pmb0$, $\pmb\Sigma$ the $6\times 6$ identity matrix, and $\nu=\nu_0+(\delta-4)$ for $\delta=1,\ldots,7$. This choice permits us to test the power of our test against both higher and lower values of the parameter of interest. Since $\nu_0>4$, we estimate the nuisance parameters $\pmb\mu$ and $\pmb\Sigma$ by means of the sample mean and sample covariance matrix, respectively; this is why we denote our test  $\phi^{(n)}_{MeCov}$. The results at the $5\%$ level are reported in Table 1. We clearly see that the two tests always detect the deviation from the null hypothesis, and that the differences in performance logically shrink when the sample size increases. Quite interestingly, it appears that our test performs better against right-sided alternatives (which may be explained by the fact that the sample mean and covariance estimators become more efficient when we approach the multinormality situation), whereas the LR test is better for left-sided alternatives. 

One major advantage of our approach lies in the fact that any root-$n$ consistent estimators $\hat{\pmb\mu}^{(n)}$ and $\hat{\pmb\Sigma}^{(n)}$ can be used to estimate the unknown nuisance parameters. In order to detect the finite-sample effect  of distinct such estimators, we have conducted a further simulation study with $N=2,500$ replications, this time in dimension $k=2$, for $\pmb\mu=\pmb0$, $\pmb\Sigma$  with entries $\Sigma_{11}=5, \Sigma_{21}= \Sigma_{12}=3$ and $\Sigma_{22}=2$, and $\nu_0=5$ against $\nu=2,3,4,6,7,8$.  Besides $\phi^{(n)}_{MeCov}$ we have  used the tests $\phi^{(n)}_{MOT}$,  based on the M\"ott\"onen-Oja spatial median and the adjusted Tyler shape matrix, and $\phi^{(n)}_{MCD}$, based on the MCD estimators for location and scatter. The results at the $5\%$ level are reported in Table 2. As could be expected, the latter two tests perform less well against right-sided alternatives (MCD-based methods are known to exhibit very moderate power); they even exhibit very low performances for $n=200$. They outperform  $\phi^{(n)}_{MeCov}$ and the LR test against left-sided alternatives, even for $n=200$ (it is however to be remarked that for this sample size the MCD-based test lies slightly above the nominal level constraint); a reason for this improved performance compared to $\phi^{(n)}_{MeCov}$ are the good robustness properties of the MOT and MCD estimators. Again, these differences are waning with the sample size.

\section{Real-data example}\label{sec:real}

In this section, we  apply our optimal test on a real-data example, namely on a data set made of 9 years of daily returns of 22 major worldwide market indexes that represent three geographical areas: America (S\&P500, NASDAQ, TSX, Merval, Bovespa and IPC), Europe and Middle East (AEX, ATX, FTSE, DAX, CAC40, SMI, MIB and TA100), and East Asia and Oceania (HgSg, Nikkei, StrTim, SSEC, BSE, KLSE, KOSPI and AllOrd). The sample consists of 2536 observations, from January 4, 2000 to September 22, 2009. The same data has already been analyzed in Dominicy and Veredas (2013) and Dominicy  \mbox{\emph{et al.}} (2013).  We refer the reader to Table 5 of Dominicy and Veredas~(2013) for some information about the descriptive statistics of each return series, and to Dominicy  \mbox{\emph{et al.}} (2013) for a detailed description of the filtering used and the data set's criticisms. The purpose of our study here is to provide  confidence intervals for the tail index under the Student $t$ assumption using the  methods developed in the present paper.

Using our right- and left-sided tests $\phi^{(n)}_{MeCov}$ at the $95\%$ confidence level, we obtain the interval $[9.44,10.96]$, while a $99\%$ confidence interval corresponds to $[9.13,11.29]$. Both intervals contain and hence confirm the value of $9.51$ derived  in Dominicy \mbox{\emph{et al.}}~(2013). Our analysis is finer in the sense that we do not provide a single estimated value but rather a confidence interval. Moreover, our tests can easily be used for testing any given value of the null hypothesis, e.g. the value $\nu=8$ is dramatically rejected by our right-sided test with a p-value of $6.8\times 10^{-7}$.

\section{Final comments and outlook on efficient estimation methods}\label{sec:est}

We have proposed in this paper a new efficient way, based on the Le Cam methodology, to tackle hypothesis testing problems about the tail weight parameter in multivariate Student $t$ distributions. Our tests are as powerful as the classical procedures based on maximum likelihood estimation, but improve on the latter by their simpler form and by avoiding the fallacies inherent to ML estimation under the $t$ model (see the Introduction). Moreover our methodology allows us to calculate explicit asymptotic power expressions against sequences of contiguous alternatives. Therefore, we conceive our tests as attractive alternatives to the classical procedures, all-the-more since the practitioners can freely choose their favorite root-$n$ consistent estimators for the nuisance parameters.

Yet another advantage of the methodology developed in the present paper lies in its applicability in point estimation. More precisely, using the ULAN structure and the efficient central sequence for tail weight, our framework readily leads to the construction of so-called optimal \emph{one-step estimators} for~$\nu$ (of course also for $\pmb\mu$ and $\pmb\Sigma$, but we here keep our focus on $\nu$). The main idea behind one-step estimation consists in adding to an existing adequate preliminary estimator $\hat{\nu}^{(n)}$ a quantity depending on a version of the efficient central sequence for $\nu$. More precisely, the one-step estimator takes on the guise
\begin{align}\label{onestep} \hat{\nu}_{Cam}^{(n)}=\hat{\nu}^{(n)}+ n^{-1/2}\left(\Gamma_{33}^*\left(\pmb{\hat{\mu}}^{(n)}, \pmb{\hat{\Sigma}}^{(n)}, \hat{\nu}^{(n)}\right)\right)^{-1}\Delta_3^{(n)*}\left( \pmb{\hat{\mu}}^{(n)}, \pmb{\hat{\Sigma}}^{(n)}, \hat{\nu}^{(n)}\right),\end{align}
where $\left(\pmb{\hat{\mu}}^{(n)}, \pmb{\hat{\Sigma}}^{(n)}, \hat{\nu}^{(n)}\right)$ is a preliminary estimator of $\left(\pmb{\mu}, \pmb{\Sigma},\nu \right)$ fulfilling Assumption~A. The following result states the asymptotic properties of the one-step estimator. 

\begin{theorem}\label{prop}
Suppose that $\hat{\pmb{\mu}}^{(n)}$ and $\hat{\pmb{\Sigma}}^{(n)}$ satisfy Assumption~A. Let $\hat{\nu}^{(n)}$ be an estimator of $\nu$ fulfilling also Assumption A and let $\hat{\nu}^{(n)}_{Cam}$ be the one-step estimator given by (\ref{onestep}). Then, under ${\rm P}^{(n)}_{\pmb{\pmb{\mathcal{V}}}}$,  
$$n^{1/2}(\hat{\nu}^{(n)}_{Cam}-\nu) \stackrel{\mathcal{L}} \to \mathcal{N}\left( 0, (\Gamma_{33}^*(\pmb{\mathcal{V}}))^{-1}\right)$$
as $n \to \infty$. Moreover, $\hat{\nu}^{(n)}_{Cam}$ is the most efficient estimator for $\nu$ under $\bigcup_{\pmb{\mu} \in \mathbb{R}^k} \bigcup_{\pmb{\Sigma}\in S_{k}}\bigcup_{\nu>0}{\rm P}^{(n)}_{(\pmb{\mu}, \pmb{\Sigma}, \nu)}$.
\end{theorem}

The short proof is provided in the appendix. Theorem \ref{prop} shows that, whatever the performance of the preliminary root-$n$ consistent estimator $\hat{\nu}^{(n)}$ and whatever the choice of root-$n$ consistent estimators for location and scatter, the one-step estimator $\hat{\nu}^{(n)}_{Cam}$ is as efficient as the maximum likelihood estimator (MLE).  Hence we can make any root-$n$ consistent estimator for tail weight as efficient as the MLE in a quite simple way. Obvious examples are the Method of Moments estimator (whose multivariate version is given in Mardia,~1970) or, if moments are to be avoided (moment-based estimators require $\nu>8$ to be root-$n$ consistent), estimators based on quantiles of the $t$-based Mahalanobis distance. Thus Theorem~\ref{prop} opens the door to numerous possibilities for efficiently estimating $\nu$ without having recourse to maximum likelihood estimation. A detailed study of these estimators goes beyond the scope of the present paper and is therefore left for future research.

\begin{table} 
\begin{center}
\begin{tabular}{|l|ccccccc|}
\hline
$k=6$&\multicolumn{7}{c}{$n=200$}\vline\\\hline
Test&$\nu=2$&$\nu=3$&$\nu=4$& $\nu_0=5$&$\nu=6$&$\nu=7$&$\nu=8$\\\hline
LR& 1.0000 &0.9432& 0.3092& 0.0536   &0.2208 &0.5268 &0.7940\\
$\phi^{(n)}_{MeCov}$ & 1.0000 &0.9000 & 0.2116 &0.0568  & 0.2504 &0.5612& 0.8172  \\
\hline
Test&$\nu=7$&$\nu=8$&$\nu=9$& $\nu_0=10$&$\nu=11$&$\nu=12$&$\nu=13$\\\hline
LR& 0.3632 &0.1616& 0.0712 &0.0592& 0.0976& 0.1588 &0.2260\\
$\phi^{(n)}_{MeCov}$ & 0.2616 &0.1000 &0.0468 &0.0540& 0.1052 &0.1840 &0.2532  \\
\hline
&\multicolumn{7}{c}{$n=500$}\vline\\\hline
Test&$\nu=2$&$\nu=3$&$\nu=4$& $\nu_0=5$&$\nu=6$&$\nu=7$&$\nu=8$\\\hline
LR& 1.0000 &    1.0000& 0.6456 &0.0484 &0.4428& 0.8828 &0.9944\\
$\phi^{(n)}_{MeCov}$ & 1.0000 &1.0000& 0.5844 &0.0504 &0.4732& 0.8908& 0.9948  \\
\hline
Test&$\nu=7$&$\nu=8$&$\nu=9$& $\nu_0=10$&$\nu=11$&$\nu=12$&$\nu=13$\\\hline
LR& 0.7472& 0.3508 &0.1036& 0.0548 &0.1116 &0.2336& 0.4368\\
$\phi^{(n)}_{MeCov}$ & 0.6956 &0.2932 &0.0840 &0.0572 &0.1280& 0.2584 &0.4600  \\
\hline
\end{tabular}
\caption{Rejection frequencies (out of $N=2,500$ replications), under $6$-variate Student $t$ densities with location $\pmb{\mu}=\pmb0$ and  scatter matrix ${\pmb\Sigma}={\pmb I}_6$, the $6\times 6$ identity matrix, for the testing problems $\nu_0=5$ against $\nu=2,3,4,6,7,8$ and $\nu_0=10$ against $\nu=7,8,9,11,12,13$, of the likelihood ratio test LR and the optimal test $\phi_{MeCov}^{(n)}$ for the sample sizes $n=200$ and 500. The nominal  level is $\alpha=0.05$.}\label{tab2}
\end{center}
\end{table}

\begin{table} 
\begin{center}
\begin{tabular}{|l|ccccccc|}
\hline
$k=2$&\multicolumn{7}{c}{$n=200$}\vline\\\hline
Test&$\nu=2$&$\nu=3$&$\nu=4$& $\nu_0=5$&$\nu=6$&$\nu=7$&$\nu=8$\\\hline
LR& 0.9972& 0.6708 &0.1692 &0.0508 &0.1200& 0.2328 &0.3544\\
$\phi^{(n)}_{MeCov}$& 0.9964 &0.6576& 0.1724& 0.0464 &0.1016 &0.1880& 0.3044\\
$\phi^{(n)}_{MOT}$ & 0.9988 &0.7884 &0.2772 &0.0568 &0.0308 &0.0664& 0.1200 \\
$\phi^{(n)}_{MCD}$  & 0.9992 &0.8092& 0.3116& 0.0716 &0.0356 &0.0512 &0.1084 \\
\hline
&\multicolumn{7}{c}{$n=500$}\vline\\\hline
Test&$\nu=2$&$\nu=3$&$\nu=4$& $\nu_0=5$&$\nu=6$&$\nu=7$&$\nu=8$\\\hline
LR& 1.0000 &0.9720& 0.3220 &0.0556 &0.2012 &0.4736 &0.7080\\
$\phi^{(n)}_{MeCov}$& 1.0000& 0.9716 &0.3344 &0.0532 &0.1796& 0.4468& 0.6728\\
$\phi^{(n)}_{MOT}$ & 1.0000& 0.9844 &0.4168 &0.0500& 0.0984 &0.3188 &0.5420\\
$\phi^{(n)}_{MCD}$  & 1.0000& 0.9864& 0.4280& 0.0524 &0.0876& 0.2764 &0.4944 \\
\hline
\end{tabular}
\caption{Rejection frequencies (out of $N=2,500$ replications), under $2$-variate Student $t$ densities with location $\pmb{\mu}=\pmb0$ and scatter $\pmb\Sigma$ with $\Sigma_{11}=5,\Sigma_{21}=\Sigma_{12}=3$ and $\Sigma_{22}=2$, for the null hypothesis $\nu_0=5$ against $\nu=2,3,4,6,7,8$, of the likelihood ratio test LR and the optimal tests $\phi^{(n)}_{MeCov}$, $\phi^{(n)}_{MOT}$ and $\phi^{(n)}_{MCD}$.   The nominal  level is $\alpha=0.05$.}\label{tab2}
\end{center}
\end{table}

{
\appendix

\section{Technical proofs.}

\noindent {\bf Proof of Lemma 1}. It follows from the asymptotic linearity property of the central sequence given in (\ref{uniform}) that, under  ${\rm P}_{\pmb{\mathcal{V}}}^{(n)}$ and for $n\rightarrow\infty$,
\begin{align}\label{delta3}\Delta_3^{(n)}\left(\pmb{\mathcal{V}}+n^{-1/2}\pmb{\tau}^{(n)}\right)=\Delta_3^{(n)}(\pmb{\mathcal{V}})-(\pmb{\Gamma}_{23}(\pmb{\mathcal{V}}))'\pmb{\tau}_2^{(n)}-\Gamma_{33}(\nu)\tau_3^{(n)}+o_{{\rm P}}(1),\end{align}
and 
\begin{align}\label{delta2}\pmb{\Delta}_2^{(n)}\left(\pmb{\mathcal{V}}+n^{-1/2}\pmb{\tau}^{(n)}\right)=\pmb{\Delta}_2^{(n)}(\pmb{\mathcal{V}})-\pmb{\Gamma}_{22}(\pmb{\mathcal{V}})\pmb{\tau}^{(n)}_2-\pmb{\Gamma}_{23}(\pmb{\mathcal{V}})\tau_3^{(n)}+o_{{\rm P}}(1).\end{align}
By the definition of the central sequence for tail weight, we have
\begin{align} \label{deltastar}
\Delta_3^{(n)*}\left(\pmb{\mathcal{V}}+n^{-1/2}\pmb{\tau}^{(n)}\right)&=\Delta_3^{(n)}\left(\pmb{\mathcal{V}}+n^{-1/2}\pmb{\tau}^{(n)}\right)-\left(\pmb{\Gamma}_{23}\left(\pmb{\mathcal{V}}+n^{-1/2}\pmb{\tau}^{(n)}\right)\right)' \left(\pmb{\Gamma}_{22}\left(\pmb{\mathcal{V}}+n^{-1/2}\pmb{\tau}^{(n)}\right)\right)^{-1}\nonumber\\
& \qquad \times \pmb{\Delta}_2^{(n)}\left(\pmb{\mathcal{V}}+n^{-1/2}\pmb{\tau}^{(n)}\right).
\end{align}
Substituting (\ref{delta3}) and (\ref{delta2}) in (\ref{deltastar}) yields
\begin{align*}
&\Delta_3^{(n)*}\left(\pmb{\mathcal{V}}+n^{-1/2}\pmb{\tau}^{(n)}\right)\\
=\quad&\Delta_3^{(n)*}(\pmb{\mathcal{V}})-\left[ \left(\pmb{\Gamma}_{23}\left(\pmb{\mathcal{V}}+n^{-1/2}\pmb{\tau}^{(n)}\right)\right)'\left(\pmb{\Gamma}_{22}\left(\pmb{\mathcal{V}}+n^{-1/2}\pmb{\tau}^{(n)}\right)\right)^{-1}-(\pmb{\Gamma}_{23}(\pmb{\mathcal{V}}))'(\pmb{\Gamma}_{22}(\pmb{\mathcal{V}}))^{-1}\right]\pmb{\Delta}_2^{(n)}(\pmb{\mathcal{V}})\\&-\left(\pmb{\Gamma}_{23}(\pmb{\mathcal{V}})\right)'\pmb{\tau}_2^{(n)}-\Gamma_{33}(\nu)\tau_3^{(n)}+\left(\pmb{\Gamma}_{23}\left(\pmb{\mathcal{V}}+n^{-1/2}\pmb{\tau}^{(n)}\right)\right)'\left(\pmb{\Gamma}_{22}\left(\pmb{\mathcal{V}}+n^{-1/2}\pmb{\tau}^{(n)}\right)\right)^{-1}\pmb{\Gamma}_{22}(\pmb{\mathcal{V}})\pmb{\tau}_2^{(n)}\\&
+\left(\pmb{\Gamma}_{23}\left(\pmb{\mathcal{V}}+n^{-1/2}\pmb{\tau}^{(n)}\right)\right)'\left(\pmb{\Gamma}_{22}\left(\pmb{\mathcal{V}}+n^{-1/2}\pmb{\tau}^{(n)}\right)\right)^{-1}\pmb{\Gamma}_{23}(\pmb{\mathcal{V}})\tau_3^{(n)}+o_{{\rm P}}(1),
\end{align*}
under  ${\rm P}_{\pmb{\mathcal{V}}}^{(n)}$ as $n\rightarrow\infty$. Hence, the result follows from the continuity of $\pmb{\mathcal{V}} \mapsto \pmb{\Gamma}(\pmb{\mathcal{V}})$ and  the boundedness in probability of the central sequence. \hfill $\square$
 \vspace{0.2cm}

\noindent {\bf Proof of Theorem 2}. Since $\Gamma_{33}^*(\pmb{\mathcal{V}})$ is continuous in both $\pmb{\mu}$ and $\pmb{\Sigma}$, we readily have for any bounded sequence $\left(\pmb{\tau}_1^{(n)}, \pmb{\tau}_2^{(n)}\right) \in \mathbb{R}^k \times M_k$ that $\lim_{n\rightarrow\infty}\Gamma_{33}^*\left(\pmb{\mu}+n^{-1/2}\pmb{\tau}^{(n)}_1, \pmb{\Sigma}+n^{-1/2}\pmb{\tau}^{(n)}_2, \nu_0\right)= \Gamma_{33}^*(\pmb{\mu}, \pmb{\Sigma}, \nu_0)$. Since this convergence of course implies convergence in probability, Lemma~4.4  in Kreiss (1987) allows us to replace the non-random quantities with root-$n$ consistent and locally and asymptotically discrete estimators. Hence, Slutsky's Lemma, combined with Lemma \ref{linea}, entails that under ${\rm P}^{(n)}_{(\pmb{\mu}, \pmb{\Sigma}, \nu_0)}$,
\begin{equation}\label{contig}
Q^{(n)}_{\nu_0}=\frac{\Delta_3^{(n)*}\left(\pmb{\hat{\mu}}^{(n)}, \pmb{\hat{\Sigma}}^{(n)}, \nu_0\right)}{\sqrt{\Gamma_{33}^*\left(\pmb{\hat{\mu}}^{(n)}, \pmb{\hat{\Sigma}}^{(n)}, \nu_0\right)}}=\frac{\Delta_3^{(n)*}(\pmb{\mu}, \pmb{\Sigma}, \nu_0)}{\sqrt{\Gamma_{33}^*(\pmb{\mu}, \pmb{\Sigma}, \nu_0)}}+o_{\rm P}(1)
\end{equation}
as $n \to \infty$. The proof of the statement in Part (i) then follows, since $\Delta_3^{(n)*}\left(\pmb{\mu}, \pmb{\Sigma}, \nu_0\right)$ is asymptotically $\mathcal{N}(0, \Gamma_{33}^*(\pmb{\mu}, \pmb{\Sigma}, \nu_0))$ under $\bigcup_{\pmb{\mu} \in \mathbb{R}^k} \bigcup_{\pmb{\Sigma}\in S_{k}}{\rm P}^{(n)}_{(\pmb{\mu}, \pmb{\Sigma}, \nu_0)}$ by the Central Limit Theorem. Moreover, still under ${\rm P}^{(n)}_{(\pmb{\mu}, \pmb{\Sigma}, \nu_0)}$ and for any bounded sequence $\pmb{\tau}^{(n)}=\left(\pmb{\tau}_1^{(n)}, \pmb{\tau}_2^{(n)}, \tau_3^{(n)}\right) \in \mathbb{R}^k \times M_k \times \mathbb{R}$, we see that, as $n\rightarrow\infty$, 
$$\begin{pmatrix} 
\Delta_3^{(n)*}(\pmb{\mu}, \pmb{\Sigma}, \nu_0) \\
\Lambda^{(n)}_{(\pmb{\mu}, \pmb{\Sigma}, \nu_0)+n^{-1/2}\pmb{\tau}^{(n)}/(\pmb{\mu}, \pmb{\Sigma}, \nu_0)} 
\end{pmatrix} \stackrel{\mathcal{L}}{\longrightarrow}\mathcal{N}_2 \left( \begin{pmatrix} 0 \\ -\frac{1}{2} \mathcal{T}'\pmb{\Gamma}(\pmb{\mu}, \pmb{\Sigma}, \nu_0)\mathcal{T}\end{pmatrix}, \begin{pmatrix} \Gamma_{33}^*(\pmb{\mu}, \pmb{\Sigma}, \nu_0) & \tau_3\Gamma_{33}^*(\pmb{\mu}, \pmb{\Sigma}, \nu_0)\\
\tau_3 \Gamma_{33}^*(\pmb{\mu}, \pmb{\Sigma}, \nu_0) & \mathcal{T}'\pmb{\Gamma}(\pmb{\mu}, \pmb{\Sigma}, \nu_0)\mathcal{T}\end{pmatrix}\right),$$
where $\Lambda^{(n)}_{(\pmb{\mu}, \pmb{\Sigma}, \nu_0)+n^{-1/2}\pmb{\tau}^{(n)}/(\pmb{\mu}, \pmb{\Sigma}, \nu_0)} $ is the log-likelihood ratio and $\mathcal{T}^{(n)}:=((\pmb{\tau}_1^{(n)})',({\rm vech}(\pmb\tau_2^{(n)}))',\tau_3^{(n)})'$, with $\mathcal{T}=\lim_{n \to \infty} \mathcal{T}^{(n)}$. Le Cam's third lemma thus implies that $\Delta_3^{(n)*}(\pmb{\mu}, \pmb{\Sigma}, \nu_0)$ is asymptotically
$\mathcal{N}\left(\tau_3\Gamma_{33}^*(\pmb{\mu}, \pmb{\Sigma}, \nu_0), \Gamma_{33}^*(\pmb{\mu}, \pmb{\Sigma}, \nu_0)\right)$ under $\bigcup_{\pmb{\mu} \in \mathbb{R}^k} \bigcup_{\pmb{\Sigma}\in S_{k}}{\rm P}^{(n)}_{(\pmb{\mu}, \pmb{\Sigma}, \nu_0+n^{-1/2}\tau_3^{(n)})}$. Since (\ref{contig}) holds as well under ${\rm P}^{(n)}_{(\pmb{\mu}, \pmb{\Sigma}, \nu_0+n^{-1/2}\tau_3^{(n)})}$ by contiguity, Part (ii) of the theorem readily follows.

As regards Part (iii), the fact that $\phi^{(n)}_{\nu_0}$ has  asymptotic level $\alpha$ follows directly from the asymptotic null distribution given in Part (i), while local asymptotic maximinity is a consequence of the weak convergence of the local experiments to the Gaussian shift experiment. \hfill $\square$
 \vspace{0.2cm}

\noindent {\bf Proof of Theorem 3}. Let us start by showing that   the asymptotic distribution under ${\rm P}^{(n)}_{\pmb{\mathcal{V}}}$ of $n^{1/2}(\hat{\nu}^{(n)}_{Cam}-\nu)$
is the same as that of $(\Gamma_{33}^*(\pmb{\mathcal{V}}))^{-1}\Delta_3^{(n)*}(\pmb{\mathcal{V}})$.  From the asymptotic linearity of $\Delta_3^{(n)*}(\pmb{\mathcal{V}})$ in Lemma~\ref{linea} combined with Lemma~4.4 of Kreiss~(1987) under Assumption~A and from the continuity of $\pmb{\mathcal{V}} \mapsto \Gamma_{33}^*(\pmb{\mathcal{V}})$, we obtain that
\begin{align*}
n^{1/2}(\hat{\nu}^{(n)}_{Cam}-\nu)& = n^{1/2}(\hat{\nu}^{(n)}-\nu)+\left(\Gamma_{33}^*\left(\pmb{\hat{\mu}}^{(n)}, \pmb{\hat{\Sigma}}^{(n)}, \hat{\nu}^{(n)}\right)\right)^{-1}\Delta_3^{(n)*}\left( \pmb{\hat{\mu}}^{(n)}, \pmb{\hat{\Sigma}}^{(n)}, \hat{\nu}^{(n)}\right)\\
&=  n^{1/2}(\hat{\nu}^{(n)}-\nu)+(\Gamma_{33}^*(\pmb{\hat{\mu}}^{(n)}, \pmb{\hat{\Sigma}}^{(n)}, \hat{\nu}^{(n)}))^{-1}\left(\Delta_3^{(n)*}(\pmb{\mathcal{V}})-n^{1/2}(\hat{\nu}^{(n)}-\nu)\Gamma_{33}^*(\pmb{\mathcal{V}})\right)+o_{\rm P}(1)\\
&= (\Gamma_{33}^*(\pmb{\mathcal{V}}))^{-1}\Delta_3^{(n)*}(\pmb{\mathcal{V}})+o_{\rm P}(1)
\end{align*}
under ${\rm P}^{(n)}_{\pmb{\mathcal{V}}}$ as $n \to \infty$. The asymptotic behavior then directly follows thanks to the ULAN property. Efficiency of the estimator can be seen by noticing that the asymptotic variance coincides with the inverse of the efficient Fisher information. \hfill $\square$
 \vspace{0.2cm}}

\section*{Acknowledgments}
\noindent The research of Christophe Ley, who is also a member of ECARES, is supported by a Mandat de Charg\'e de Recherche from the Fonds National de la Recherche Scientifique, Communaut\'e fran\c caise de Belgique, whereas that of Anouk Neven is supported by the Fonds National de la Recherche, Luxembourg (Project Reference 4086487). Both authors would like to thank Yves Dominicy and David Veredas for providing them with the financial data used in Section \ref{sec:real} and Giovanni Peccati and Enrique Sentana for interesting comments and discussions.

\end{document}